\documentclass[10pt]{llncs}
\usepackage{tikz}
\usetikzlibrary{automata,positioning}
\usetikzlibrary{arrows,automata}

\usepackage{amsmath,latexsym}
\usepackage{amsfonts}
\usepackage{amssymb}
\usepackage{amssymb}
\usepackage{graphicx}
\usepackage{epstopdf}
\usepackage{subfigure}
\usepackage{setspace}
\usepackage{authblk}
\usepackage{algpseudocode}
\usepackage{algorithm}

\newtheorem{thrm}{Theorem}
\newtheorem{lemm}[thrm]{Lemma}
\newtheorem{defn}[thrm]{Definition}
\newtheorem{examp}[thrm]{Example}
\newtheorem{obs}[thrm]{Observation}

\newenvironment{my_itemize}{
\begin{itemize}
  \setlength{\itemsep}{1pt}
  \setlength{\parskip}{0pt}
  \setlength{\parsep}{0pt}}{\end{itemize}
}

\def\qedbox#1#2{\vbox{\hrule height.2pt
  \hbox{\vrule width.2pt height#2pt \kern#1pt \vrule width.2pt}
  \hrule height.2pt}}
\def\qed{\hfill \quad\qedbox46\newline}

\def\dd{\mathinner{\ldotp\ldotp}}
\def\s#1{\mbox{\boldmath $#1$}}

\def\+{\!+\!}
\def\-{\!-\!}

\def\m{\!-\!}

\def\itbf#1{\textit{\textbf{#1}}}

\def\ACTIVE{\mbox{ACTIVE}}

\def\peek{\bf{peek}}

\def\bproc{{\bf procedure\ }}

\def\bfor{{\bf for\ }}
\def\bto{{\bf to\ }}

\def\bwhile{{\bf while\ }}

\def\band{{\bf and\ }}
\def\bor{{\bf or\ }}
\def\bdo{{\bf do\ }}
\def\bif{{\bf if\ }}
\def\bthen{{\bf then\ }}
\def\belse{{\bf else\ }}

\def\breturn{{\bf return\ }}

\def\la{\leftarrow}

\def\q{\quad}
\def\qq{\qquad}
\def\com#1{{\bf $\triangleright$}\hspace{6pt}{\sl #1}}

\def\pref(#1,#2){$#1$ is a prefix of $#2$}
\def\suff(#1,#2){$#1$ is a suffix of $#2$}

\def\SA{\mbox{SA}}
\def\ISA{\mbox{ISA}}
\def\offset{\mbox{offset}}
\def\prev{\mbox{prev}}
\def\period{\mbox{period}}
\def\per{\textbf{period}}
\def\nextequal{\mbox{nextequal}}
\def\nexteq{\textbf{nextequal}}
\def\reg(#1,#2){$#2$ is $#1$-regular}
\def\notreg(#1,#2){$#2$ is not $#1$-regular}

\def\UPDATE\_F{\tt{UPDATE\_F}}

\def\L{\mathcal{L}}

\def\ACTIVE{\rm{ACTIVE}}
\def\MVR*{\rm{MVR^*}}
\def\MVR{\rm{MVR}}

\def\bpush{{\bf push}}
\def\bpeek{{\bf peek}}
\def\bpop{{\bf pop}}
\def\O{\mathcal{O}}
\def\MATCH{\rm{MATCH}}
\def\EXIT{\rm{EXIT}}
\def\COMP{\rm{COMP}}
\def\NSV{\mbox{NSV}}
\def\PNSV{\mbox{PNSV}}
\def\NPNSV{\mbox{NPNSV}}

\newif\ifProofs
\Proofsfalse

\newif\ifRev
\Revfalse

\begin{document}

\pagestyle{headings}
\title{Algorithms to Compute the Lyndon Array\thanks
{This work was supported in part by the
Natural Sciences \& Engineering Research Council of Canada.
The authors wish to thank Maxime Crochemore and Hideo Bannai
for helpful discussions.}
}
\author{Frantisek Franek\inst{1}
\and
A.\ S.\ M.\ Sohidull Islam\inst{2}
\and M.\ Sohel Rahman\inst{3}
\and \\ 
W.\ F.\ Smyth\inst{1,3,4}
}
\institute{$\!^1\ $Algorithms Research Group \\ 
Department of Computing \& Software \\ 
McMaster University, Hamilton, Canada \\
\email{$\{$franek/smyth$\}$@mcmaster.ca} \\ 
{\ } \\ 
$\!^2\ $School of Computational Science \& Engineering \\ 
McMaster University, Hamilton, Canada \\ 
\email{sohansayed@gmail.com} \\ 
{\ } \\ 
$\!^3\ $Department of Computer Science \& Engineering \\
Bangladesh University of Engineering \& Technology \\
\email{msrahman@cse.buet.ac.bd} \\
{\ } \\ 
$\!^4\ $School of Engineering \& Information Technology \\
Murdoch University, Perth, Australia
}

\maketitle

\begin{abstract}
In the Lyndon array $\s{\s{\lambda}} = \s{\s{\lambda}}_{\s{x}}[1..n]$ of a string $\s{x} = \s{x}[1..n]$,
$\s{\lambda}[i]$ is the length of the longest Lyndon word starting at
position $i$ of \s{x}.
The computation of $\s{\lambda}$ has recently become of
great interest, since it was shown
(Bannai {\it et al.}, {\bf The ``Runs'' Theorem} \cite{BIINTT14})
that the runs in \s{x} are computable in linear time
from $\s{\lambda_x}$.
Here we first describe three algorithms for computing $\s{\lambda_x}$
that have been suggested in the literature,
but for which no structured exposition has been given.
Two of these algorithms execute in $\O(n^2)$ time in the worst case;
the third achieves $\Theta(n)$ time,
but at the expense of prior computation of both the suffix array
and the inverse suffix array of \s{x}.
We then go on to describe two variants of a new algorithm that
avoids prior computation of global data structures
and executes in worst-case $\O(n\log n)$ time.
Experimental evidence suggests that all but one of these five
algorithms require only linear execution time in practice,
with the two new algorithms faster by a small factor.
% of which one, a recursive application of Duval's Lyndon factorization algorithm,
% is shown experimentally to be supralinear in the average case,
% while requiring $\O(n\log n)$ time in the worst case
% on binary strings.
% Preliminary experiments indicate that the other three algorithms
% execute in $\Theta(n)$ time on average,
% with no single algorithm showing a clear advantage over the others.
% One of them
% achieves $\Theta(n)$ worst-case time
% based on ``Next Smaller Value'' (NSV) methodology applied to
% the inverse suffix array $\ISA_{\s{x}}$,
% but seems to be somewhat slow in practice.
% Another is a ``folklore'' algorithm,
% based on an interesting connection to the prefix table,
% that, although requiring $\O(n^2)$ time in the worst case,
% is generally very fast.
% Finally, two variants of a new algorithm are described,
% again using NSV techniques,
% one worst-case $\O(n^2\log n)$,
% the other worst-case $\O(n\log n)$,
% but with very similar run times in the tests conducted so far.
We conjecture that there exists a fast
and worst-case linear-time algorithm to compute the Lyndon array
that is also ``elementary''
(making no use of global data structures such as the suffix array).
\end{abstract}

\section{Introduction}
\label{sect-intro}
If $\s{x} = \s{uv}$ for some \s{u} and nonempty \s{v},
then \s{vu} is said to be the $|\s{u}|^{\mbox{th}}$ \itbf{rotation} of \s{x},
written $\s{vu} = R_{|\s{u}|}(x)$.
If there exists a string \s{u} and an integer $e > 1$
such that $\s{x} = \s{u}^e$,
then \s{x} is said to be a \itbf{repetition};
otherwise \s{x} is \itbf{primitive}.
A primitive string \s{x} that is
lexicographically least among all its rotations
$R_k(\s{x}), k= 0,1,\ldots,|\s{x}|\- 1$,
is said to be a \itbf{Lyndon word}.

The \itbf{Lyndon array} $\s{\lambda} = \s{\lambda}_{\s{x}}[1..n]$
(equivalently, $\L = \L_{\s{x}}[1..n]$)
of a given nonempty string $\s{x} = \s{x}[1..n]$
gives at each position $i$ the length
(equivalently, the end position)
of the longest Lyndon word starting at $i$:

\begin{equation}
\label{ex1}
\begin{array}{rccccccccc}
\scriptstyle 1 & \scriptstyle 2 & \scriptstyle 3 & \scriptstyle 4 & \scriptstyle 5 & \scriptstyle 6 & \scriptstyle 7 & \scriptstyle 8 & \scriptstyle 9 & \scriptstyle 10 \\
\s{x} = a & b & a & a & b & a & b & a & a & b \\
\s{\lambda} = 2 & 1 & 5 & 2 & 1 & 2 & 1 & 3 & 2 & 1 \\
\L = 2 & 2 & 7 & 5 & 5 & 7 & 7 & 10 & 10 & 10
\end{array}
\end{equation}
The Lyndon array has recently become of interest since
Bannai {\it et al.} \cite{BIINTT14} showed that it could
be used to efficiently compute all the maximal periodicities (``runs'')
in a string.
In this paper we describe four algorithms to compute $\s{\lambda}_{\s{x}}$,
three of them shown experimentally to be running in $\Theta(n)$ time in practice.
Section~\ref{sect-prelim} makes various observations that
apply generally to the Lyndon array and its computation.
In Section~\ref{sect-folklore} we describe three
algorithms, two that require $\O(n^2)$ time in the worst case,
of which one is very fast and apparently linear in practice,
the other supralinear in practice and $\O(n\log n)$ in the average case on binary strings.
The third algorithm is simple and worst-case linear-time,
but requires suffix array construction and so is a little slower.
Section~\ref{sect-newalgs} describes two variants of a new algorithm
that uses only elementary data structures (no suffix arrays).
One variant is $\O(n^2)$ in the worst case,
the other guarantees $\O(n\log n)$ time,
but with no clear advantage in processing time.
Section~\ref{sect-exp} describes the results of
preliminary experiments on the algorithms;
Section~\ref{sect-future} outlines future work.

\section{Preliminaries}
\label{sect-prelim}
Here we make various observations that apply to
the algorithms described below.
\begin{obs}
\label{obs-decomp}
Let $\s{x} = \s{w_1w_2\cdots w_k}$ be the Lyndon decompostion \cite{CFL58,D83} of \s{x},
with Lyndon words $\s{w_1} \ge \s{w_2} \ge \cdots \ge \s{w_k}$.
Then every Lyndon word $\s{x}[i..\L[i]]$ of length $\s{\lambda}[i]$
is a substring of some $\s{w_h},\ h \in 1..k$.
\end{obs}
\begin{proof}
For some $h \in 1..k\- 1$, consider \s{w_h} with nonempty proper suffix \s{v_h},
and for some $t \in 1..k\- h$, consider \s{w_{h+t}} with nonempty prefix \s{u_{h+t}}.
Since \s{w_h} is a Lyndon word, $\s{w_h} < \s{v_h}$,
and by lexorder, $\s{u_{h+t}} \le \s{w_{h+t}}$.  Thus
$\s{v_h} > \s{w_h} \ge \s{w_{h+t}} \ge \s{u_{h+t}}$,
and so $\s{v_hw_{h+1}}\cdots\s{w_{h+t-1}u_{h+t}}$
cannot be a Lyndon word for any choice of $h$ or $t$.  \qed
\end{proof}
Therefore to compute $\L_{\s{x}}$ it suffices to consider separately
each distinct element \s{w_h} in the Lyndon decomposition of \s{x}.
Hence, without loss of generality suppose \s{x} is a Lyndon word
and write it in the form $\s{x_1x_2}\cdots\s{x_m}$,
where for each $r \in 1..m$, $|\s{x_r}| = \ell_r$ and
\begin{equation}
\label{range}
\s{x_r}[1] \le \s{x_r}[2] \le \cdots \le \s{x_r}[\ell_r],
\end{equation}
while for $1 \le r < m$,
\begin{equation}
\label{drop}
\s{x_r}[\ell_r] > \s{x_{r+1}}[1].
\end{equation}
We call \s{x_r} a \itbf{range} in \s{x} and the boundary between
\s{x_r} and \s{x_{r+1}} a \itbf{drop}.
We identify a position $j$ in range \s{x_r}, $1 \le j \le \ell_r$,
with its equivalent position $i$ in \s{x} by writing
$i = S_{r,j} = \sum_{r'=1}^{r-1} \ell_{r'}\+ j$.
\begin{obs}
\label{obs-endrange}
Let $i = S_{r,j}$ be a position in \s{x} that corresponds to position $j$
in range \s{x_r}.
\begin{itemize}
\item[$(a)$]
If $\s{x_r}[j] = \s{x_r}[\ell_r]$, then $\L[i] = i$.
\item[$(b)$]
Otherwise, $\L[i] = i'$, where $i'$ is the final position in some range $\s{x_{r'}},\ r' \ge r$;
that is, $i' = \sum_{s=1}^{r'} \ell_s$.
\end{itemize}
\end{obs}
\begin{proof}
(a) is an immediate consequence of (\ref{range}) and (\ref{drop}).
To prove (b), suppose that $\s{x}[i..\L[i]]$ is a maximum-length Lyndon word,
where $\L[i]$ falls within range $r'$ but $\L[i] < i'$.
Since by (\ref{range}) $\s{x}[\L(i)] \le \s{x}[\L[i]\+ 1]$,
there are two consecutive Lyndon words $\s{x}[i..\L[i]],\s{x}[\L[i]\+ 1]$
that by the Lyndon decomposition theorem \cite{CFL58}
can be merged into a single Lyndon word $\s{x}[i..\L[i]\+ 1]$.
Thus $\s{x}[i..\L[i]]$ is not maximum-length, a contradiction.  \qed
\end{proof}
We see then that if $\s{x_r}[j] < \s{x_r}[\ell_r]$,
then $\s{x_r}[j..\ell_r]$ is a (not necessarily maximum-length) Lyndon word,
and for $i = S_{r,j}$, $\L[i] \ge S_{r,\ell_r}$:
\begin{equation}
\label{ex2}
\begin{array}{rcccccccccccc}
\scriptstyle 1 & \scriptstyle 2 & \scriptstyle 3 & \scriptstyle 4 & \scriptstyle 5 & \scriptstyle 6 & \scriptstyle 7 & \scriptstyle 8 & \scriptstyle 9 & \scriptstyle 10 & \scriptstyle 11 & \scriptstyle 12 & \scriptstyle 13 \\
\s{x} = a & a & a & b\,| & a & a & b\,| & a & b\,| & a & a & b & b \\
\L = 13 & 13 & 4 & 4 & 9 & 7 & 7 & 9 & 9 & 13 & 13 & 12 & 13
\end{array}
\end{equation}

More generally, the vectors $(i,\L[i])$ satisfy a ``Monge'' property
that is exploited by Algorithm $\NSV^*$ (Section~\ref{sect-newalgs}):
\begin{obs}
\label{obs-monge}
Suppose positions $i,j$ in $\s{x}[1..n]$ satisfy $1 \le i < j \le n$.
Then either $\L[i] \le j$ or $\L[i] \ge \L[j]$:
the vectors $(i,\L[i])$ and $(j,\L[j])$ are nonintersecting.
\end{obs}
\begin{proof}
Suppose two such vectors do intersect.
Then the maximum-length Lyndon words
$\s{w_1} = \s{x}[i..\L[i]]$ and $\s{w_2} = \s{x}[j..\L[j]]$
have a nonempty overlap,
so that we can write
$\s{w_1} = \s{uv},\ \s{w_2} = \s{vv'}$ for some nonempty \s{v}.
But then, by well-known properties of Lyndon words,
$\s{w_1} < \s{v} < \s{w_2} < \s{v'}$,
implying that $\s{w_1v'}$ is a Lyndon word,
contradicting the assumption that \s{w_1} is maximum-length.  \qed
\end{proof}

Expressing a string in terms of its ranges has the same useful
lexorder property that writing it in terms of its letters does:
\begin{obs}
\label{obs-compare}
Suppose strings \s{x} and \s{y} are expressed in terms of their ranges:
$\s{x} = \s{x_1x_2}\cdots\s{x_m},\ \s{y} = \s{y_1y_2}\cdots\s{y_n}$.
Suppose further that for some least integer $r \in 1..\min(m,n)$,
$\s{x_r} \ne \s{y_r}$.
Then $\s{x} < \s{y}$ (respectively, $\s{x} > \s{y}$)
according as $\s{x_r} < \s{y_r}$ (respectively, $\s{x_r} > \s{y_r}$).
\end{obs}
\begin{proof}
If $\s{x_r} < \s{y_r}$, then either
\begin{itemize}
\item[$(a)$]
\s{x_r} is a nonempty proper prefix of \s{y_r}; or
\item[$(b)$]
there is some least position $j$ such that $\s{x_r}[j] < \s{y_r}[j]$.
\end{itemize}
In case (a), if $r = m$, then \s{x} is actually a prefix of \s{y},
so that $\s{x} < \s{y}$,
while if $r < m$, then by (\ref{drop}),
$\s{x_{r+1}}[1] < \s{y_r}[|\s{x_r}|\+ 1]$, and again $\s{x} < \s{y}$.
In case (b) the result is immediate.
The proof for $\s{x_r} > \s{y_r}$ is similar.  \qed
\end{proof}

\section{Basic Algorithms}
\label{sect-folklore}
Here we outline three algorithms
for which no clear exposition is available in the literature.
We remark that the Lyndon array computation is equivalent to ``Lyndon bracketing'',
for which an $\O(n^2)$ algorithm has been described \cite{SR03}.

\subsection{Folklore --- Iterated MaxLyn}
\label{subsect-maxlyn}
For a string \s{x} of length $n$,
recall that the \itbf{prefix table} $\pi[1..n]$
is an integer array in which for every $i \in 1\dd n$,
$\pi[i]$ is the length of the longest substring beginning at position $i$
of \s{x} that matches a prefix of \s{x}.
Given a nonempty string \s{x} on alphabet $\Sigma$, let us define $\s{x'} = \s{x}\$$,
where the sentinel $\$ < \mu$ for every letter $\mu \in \Sigma$.

\begin{obs}
\label{obs2}
\s{x} is a Lyndon word if and only if for every $i \in 2\dd n$,
$\s{x'}[1+k] < \s{x'}[i+k]$, where $k = \pi[i]$.
\end{obs}
This result forms the basis of the algorithm
given in Figure~\ref{fig-maxlyn} that computes
the length $\max \in 1\dd n-j +1$
of the longest Lyndon factor at a given position $j$ in $\s{x}[1..n]$.
Its efficiency is a consequence of the instruction
$i \la i+k+1$ that skips over positions in the range $i+1\dd i+k-1$,
effectively assuming that for every position $i^*$
in that range, $i^*+\pi[i^*] \le i\+ k$.
Lemma~\ref{lemm-maxlyn}, given in Appendix 1, justifies this assumption.
Simply repeating MaxLyn at every position $j$ of \s{x}
gives a simple, fast $\mathcal{O}(n^2)$ time and $\mathcal{O}(1)$ additional space
algorithm to compute $\s{\lambda}_{\s{x}}$.

\begin{figure}[ht]
%\leftskip=3.2cm MaxLyn$(\s{x}[1\dd n],j,\Sigma,\prec)\ :\ integer$
%
{\leftskip=2.5cm\obeylines\sfcode`;=3000
\bproc MaxLyn$(\s{x}[1\dd n],j,\Sigma,\prec)\ :\ integer$
$i \la j + 1;\ max\la 1$
\bwhile $i \le n$ \bdo
\q $k \la 0$
\q \bwhile $\s{x'}[j+k] = \s{x'}[i+k]$ \bdo
\qq $k \la k+1$
\q \bif $\s{x'}[j+k] \prec \s{x'}[i+k]$ \bthen
\qq $i \la i+k+1;\ max \la i-1$
\q \belse
\qq \breturn max
}
\caption{Algorithm MaxLyn}
\label{fig-maxlyn}
\end{figure}
Recent work on the prefix table
\cite{BKS13,CRSW15} has confirmed its importance as a data structure
for string algorithms.
In this context it is interesting to find that Lyndon words \s{x}
can be characterized in terms of $\pi_{\s{x}}$:
\begin{obs}
\label{obs-prefix}
Suppose $\s{x} = \s{x}[1\dd n]$ is a string on alphabet $\Sigma$
such that $\s{x}[1]$ is the least letter in \s{x}.
Then \s{x} is a Lyndon word over $\Sigma$ if and only if
for every $i \in 2\dd n$,
\begin{itemize}
\item[$(a)$]
$i + \pi_{\s{x}}[i] < n+1$; and
\item[$(b)$]
for every $j \in i+1\dd i+\pi_{\s{x}}[i]-1$,
$j + \pi_{\s{x}}[j] \le i + \pi_{\s{x}}[i]$.
\end{itemize}
\end{obs}

\subsection{Recursive Duval Factorization: Algorithm RDuval}
Rather than independently computing the maximum-length Lyndon factor
at each position $i$, as MaxLyn does, Algorithm RDuval
recursively computes the Lyndon decomposition into
maximum factors, at each step taking advantage of the fact that
$\L[i]$ is known for the first position $i$ in each factor,
then recomputing with the first letters removed.
By Observation~\ref{obs-decomp},
whenever $\s{x} = \s{x}[1..n]$ is a Lyndon word,
we know that $\L[1] = n$.
Thus computing the Lyndon decomposition $\s{x} = \s{w_1w_2}\cdots\s{w_k}$,
$\s{w_1} \ge \s{w_2} \ge \cdots \ge \s{w_k}$,
allows us to assign $\s{\lambda}[i_j] = |\s{w_j}|$,
where $i_j$ is the first position of \s{w_j}, $j = 1,2,\ldots,k$.

Algorithm RDuval applies this strategy recursively,
by assigning $\s{\lambda}[i_j] \la |\s{w_j}|$,
then removing the first letter $i_j$ from each $\s{w_j}$ to form $\s{w'_j}$,
to which the Lyndon decomposition is applied in the next recursive step.
This process continues until each Lyndon word is reduced to a single letter.

The asymptotic time required for RDuval is bounded above by
$n$ times the maximum depth of the recursion, thus $O(n^2)$ in the worst case --- 
consider, for example, the string $\s{x} = a^{n-1}b$.
However, to estimate expected behaviour, we can make use of a result
of Bassino {\it et al.} \cite{BassinoCN05}.
Given a Lyndon word \s{w}, they call $\s{w} = \s{uv}$
the \itbf{standard factorization} of \s{w} if \s{u} and \s{v}
are both Lyndon words and \s{v} is of maximum size.
They then show that if \s{w} is a binary string ($\Sigma = \{a,b\}$),
the average length of \s{v} is asymptotically $3|\s{w}|/4$.
Thus each recursive application of RDuval yields a left Lyndon factor
of expected length $|\s{w}|/4$
and a remainder of length $3|\s{w}|/4$ to be factored further.
It follows that the expected number of recursive calls
of RDuval is $\O(\log_{4/3} n)$.  Hence
\begin{lemm}
\label{lemm-RD}
On binary strings
RDuval executes in $O(n\log_{4/3} n)$ time on average.
\end{lemm}
\begin{examp}
For
\begin{equation*}
\label{ex1}
 \begin{array}{cccccccccccccc}
& & \scriptstyle 1 & \scriptstyle 2 & \scriptstyle 3 & \scriptstyle 4 & \scriptstyle 5 & \scriptstyle 6 & \scriptstyle 7 & \scriptstyle 8 & \scriptstyle 9 & \scriptstyle 10 & \scriptstyle 11 & \scriptstyle 12 \\ 
\s{x} & = & a & a & b & a & a & b & b & a & b & b & a & b \\ 
\s{\lambda} & = & 12 & 2 & 1 & 9 & 3 & 1 & 1 & 3 & 1 & 1 & 2 & 1
\end{array}
\end{equation*}
the factors considered are first 1--12, then
\begin{itemize}
\item[$\bullet$]
2--3 and 4--12 in the  first level of recursion;
\item[$\bullet$]
3, 5--7, 8--10 and 11--12 in the second level;
\item[$\bullet$]
$6,7,9,10,12$ in the third level.
\end{itemize}
Positions are assigned as follows:
$\s{\lambda}[1] \la 12; \s{\lambda}[2] \la 2, \s{\lambda}[4] \la 9;
\s{\lambda}[3] \la 1, \s{\lambda}[5] \la 3, \s{\lambda}[8] \la 3, \s{\lambda}[11] \la 2;
\s{\lambda}[6] \la 1, \s{\lambda}[7] \la 1, \s{\lambda}[9] \la 1,
\s{\lambda}[10] \la 1, \s{\lambda}[12] \la 1$.
\end{examp}

\subsection{NSV Applied to the Inverse Suffix Array}
\label{subsect-nsvisa}
The idea of the ``next smaller value'' (NSV) array for a given array (string) \s{x}
has been proposed in various forms and under various names
\cite{AGKR02,FMN08,OG11,GB13}.
\begin{defn}[Next Smaller Value]
Given an array $\s{x}[1..n]$ of ordered values,
$\NSV = \NSV_{\s{x}}[1..n]$ is the \itbf{next smaller value array} of \s{x}
if and only if for every $i \in 1..n$, $\NSV[i] = j$, where
\begin{itemize}
\item[$(a)$]
for every $h \in 1..j\- 1,\ \s{x}[i] \le \s{x}[i\+ h]$; and
\item[$(b)$]
either $i\+ j = n\+ 1$ or $\s{x}[i] > \s{x}[i\+ j]$.
\end{itemize}
\end{defn}
\begin{examp}
\begin{equation*}
\label{ex1}
\begin{array}{rccccccccc}
\scriptstyle 1 & \scriptstyle 2 & \scriptstyle 3 & \scriptstyle 4 & \scriptstyle 5 & \scriptstyle 6 & \scriptstyle 7 & \scriptstyle 8 & \scriptstyle 9 & \scriptstyle 10 \\ 
\s{x} = 3 & 8 & 7 & 10 & 2 & 1 & 4 & 9 & 6 & 5 \\ 
\NSV_{\s{x}} = 4 & 1 & 2 & 1 & 1 & 5 & 4 & 1 & 1 & 1
\end{array}
\end{equation*}
\end{examp}
As shown in various contexts in \cite{GB13},
$\NSV_{\s{x}}$ can be computed in $\Theta(n)$ time using a stack.
Our main observation here, touched upon in \cite{HR03},
is that $\s{\lambda}_{\s{x}}$ can be computed merely by applying NSV
to the inverse suffix array $\ISA_{\s{x}}$.
Proof of this claim can be found in Appendix 2;
here we present the very simple $\Theta(n)$-time, $\Theta(n)$-space algorithm for this calculation:
\begin{figure}[ht]
%\leftskip=2.8cm NSVISA$(\s{x}[1\dd n])\ :\ \s{\lambda}_{\s{x}}[1\dd n]$

{\leftskip=3.4cm\obeylines\sfcode`;=3000
\bproc NSVISA$(\s{x}[1\dd n])\ :\ \s{\lambda}_{\s{x}}[1\dd n]$
Compute $SA_{\s{x}}$\ \ \ \ (see \cite{NZC09,PST07})
Compute $\ISA_{\s{x}}$ from $\SA_{\s{x}}$ in place\ \ \ \ (see \cite{PST07})
$\s{\lambda}_{\s{x}} \la$ NSV$(\ISA_{\s{x}})$ (in place)
}
\caption{Apply NSV to $\ISA_{\s{x}}$}
\label{fig-3line}
\end{figure}

\section{Elementary Computation of $\s{\lambda}_{\s{x}}$ Using Ranges}
\label{sect-newalgs}

In this section we describe an approach to the computation of $\s{\lambda}_{\s{x}}$
that applies a variant of the NSV idea to the ranges of \s{x}.
Figure~\ref{nsv*} gives pseudocode for Algorithm $\NSV^*$
that uses the NSV stack ACTIVE to compute $\s{\lambda}$.
The processing identifies ranges in a single left-to-right scan
of \s{x}, making use of two range comparison routines, COMP and MATCH.
COMP compares adjacent individual ranges \s{x_r} and \s{x_{r+1}}, returning $\delta_1 = -1,0,+1$
according as $\s{x_r} < \s{x_{r+1}}$, $\s{x_r} = \s{x_{r+1}}$, $\s{x_r} > \s{x_{r+1}}$.
MATCH similarly returns $\delta_2$ for adjacent {\it sequences} of ranges; that is,
\begin{eqnarray*}
\s{X_r} &=& \s{x_rx_{r+1}\cdots x_{r+s}}, \mbox{ for some } s \ge 1; \\ 
\s{X_{r+s+1}} &=& \s{x_{r+s+1}x_{r+s+2}\cdots x_{r+s+t}}, \mbox{ for some } t \ge 1.
\end{eqnarray*}

%Added by Sohel STARTS

Algorithm $\NSV^*$ is based on the idea encapsulated in Lemma~\ref{maxlyn-suf}
of Appendix 2,
the main basis of the correctness of Algorithm NSVISA.
We process \s{x} from left to right, using a stack ACTIVE initialized with index 1.
At each iteration, the top of the stack (say, $j$)
is compared with the current index (say, $i$).
In particular, we need to compare $\s{s}_{\s{x}}(i)$ with $\s{s}_{\s{x}}(j)$,
where $\s{s}_{\s{x}}(i) \equiv \s{x}[i..n]$.
As long as $\s{s}_{\s{x}}(i) \succeq \s{s}_{\s{x}}(j)$,
$\NSV^*$ pushes the current index and continues to the next.
When $\s{s}_{\s{x}}(i) \prec \s{s}_{\s{x}}(j)$,
it pops the stack and puts
appropriate values in the corresponding indices of $\s{\lambda}_{\s{x}}$.
As noted above, especially Observations~\ref{obs-decomp}--\ref{obs-monge},
ranges are employed to expedite these suffix comparisons.

\begin{figure}[ht]
{\leftskip=1.0cm\obeylines\sfcode`;=3000
\bproc NSV* $(\s{x},\s{\lambda})$
$\nextequal \la 0^n;\ \period \la 0^n$
\bpush$(\ACTIVE) \la 1$
\com{$\s{x}[n\+ 1] = \$$, a letter smaller than any in $\Sigma$.}
\bfor $i \la 2$ \bto $n\+ 1$ \bdo
\q $\prev \la 0;\ j \la$ \bpeek$(\ACTIVE)$
\com{COMP compares suffixes specified by $i,j$ of two ranges.}
\q $\delta_1 \la \COMP(\s{x}[j],\s{x}[i]);\ \delta_2 \la 1$
\qq \bwhile ($\delta_1 \ge 0$ \band $\delta_2 > 0$) \bdo
\qq\q \bif $\delta_1 = 0$ \bthen  $\delta_2 \la \MATCH(\s{x}[j],\s{x}[i])$
\qq\q \bif $\delta_2 > 0$ \bthen
\qq\q\q \bif $\prev = 0$ \bor $\nextequal[j] \neq \prev$ \bthen $\s{\lambda}[j] \la i\m j$
\qq\q\q \belse
\qq\qq\q $\s{\lambda}[j] \la \offset \la \prev\m j$
\qq\qq\q \bif $\period[\prev]=0$ \bthen
\qq\qq\q\q \bif $\s{\lambda}[\prev] > \offset$ \bthen
%\com{The\q Lyndon word at $j$ includes the Lyndon word at $\prev$.}
\qq\qq\qq\q $\s{\lambda}[j] \la \s{\lambda}[j]\+ \s{\lambda}[\prev]$
\qq\qq\q \belse
\qq\qq\q\q \bif $\nextequal[j]=\prev$ \band $\offset \neq \s{\lambda}[\prev]$ \bthen
\qq\qq\qq\q $\s{\lambda}[j] \la \s{\lambda}[j]\+ \period[\prev]$
\qq\qq\q \bif $\s{\lambda}[\prev] = \offset$ \bthen
\qq\qq \com{Current position is a part of periodic substring}
\qq\qq\qq \bif $\period[\prev] = 0$ \bthen
\qq\qq\qq\q $\period[j] \la \period[\prev]+2\times \offset$
\qq\qq\q\q \belse
\qq\qq\qq\q $\period[j] \la \period[\prev]\+ \offset$
\qq\q\q \bpop$(\ACTIVE)$
\qq\q\q $\prev \la j;\ j \la \peek(\ACTIVE)$
\qq\qq \com{Empty stack implies termination.}
\qq\q\q \bif $j = 0$ \bthen \EXIT
\qq\q\q $\delta_1 \la \COMP(\s{x}[j],\s{x}[i])$
\com{Finished processing $i$ --- it goes to stack.}
\q \bif $\delta_2 = 0$ \bthen  $\nextequal[j] \la i$
\q \bpush$(\ACTIVE) \la i$
}
\caption{Computing $\s{\lambda}_x$ using modified NSV}
\label{nsv*}
\end{figure}

Two auxiliary arrays,
\nexteq\ and \per,
are required to handle situations in which MATCH
finds that a suffix of a previous range at position $j$
equals the current range at position $i$.
Thus, when $\delta_2 = 0$, the algorithm assigns
$\nextequal[j] \la i$ before $i$ is pushed onto ACTIVE.
Then when a later MATCH yields $\delta_2 = 0$,
the value of \per\ --- that is, the extent of the following periodicity --- may need to be set or adjusted,
as shown in the following example:

\begin{equation*}
\label{ex1}
\begin{array}{ccccccccccccccccc}
& & \scriptstyle 1 & \scriptstyle 2 & \scriptstyle 3 & \scriptstyle 4 & \scriptstyle 5 & \scriptstyle 6 & \scriptstyle 7 & \scriptstyle 8 & \scriptstyle 9 & \scriptstyle 10 & \scriptstyle 11 & \scriptstyle 12 & \scriptstyle 13 & \scriptstyle 14 & \scriptstyle 15 \\ 
\s{x} & = & a & a & a & b & a & a & b & a & a & b & a & a & b & a & b \\ 
\nextequal & = & 0 & 5 & 0 & 0 & 8 & 0 & 0 & 11 & 0 & 0 & 0 & 14 & 0 & 0 & 0 \\ 
\period & = & 0 & 12 & 0 & 0 & 9 & 0 & 0 & 6 & 0 & 0 & 0 & 4 & 0 & 0 & 0
\end{array}
\end{equation*}

A straightforward implementation of COMP and MATCH could require
a number of letter comparisons equal to the length of the shorter
of the two sequences of ranges being matched.
However, by performing $\Theta(n)$-time preprocessing,
we can compare two ranges in $\O(\sigma)$ time, where $\sigma = |\Sigma|$
is the alphabet size.
Given $\Sigma = \{\mu_1,\mu_2,\ldots,\mu_{\sigma}\}$,
we define Parikh vectors $P_r[1..\sigma]$,
where $P_r[j]$ is the number of occurrences of $\mu_j$ in range $\s{x_r}$.
Since ranges are monotone nondecreasing in the letters of the alphabet,
it is easy to compute all the $P_r, r = 1,2,\ldots,m$, in linear time
in a single scan of \s{x}.
Similarly, during the processing of each range \s{x_r},
any value $P_{r,j}$, the Parikh vector of the suffix $\s{x_r}[j..\ell_r]$,
can be computed in constant time for each position considered.
Thus we can determine the lexicographical order of any two ranges (or part ranges)
\s{x_r} and \s{x_{r'}} in $\O(\sigma)$ time
rather than time $\O(\max(\ell_r,\ell_{r'}))$.
The variant of $\NSV^*$ that uses Parikh vectors is called P$\NSV^*$;
otherwise NP$\NSV^*$ for Not Parikh.

In Appendix 3 we describe briefly another approach to
this suffix comparison problem, which also
achieves run time $\O(n \log n)$
by maintaining a simple data structure requiring $\O(n \log n)$ space.

Now consider the worst case behaviour of Algorithm $\NSV^*$.
%Added by Sohel ENDS
Given the initial string $\s{x_0} = a^hba^hc_0,\ h \ge 1$,
$c_0 > b > a$,
let $\s{x^{(h)}_k} = \s{x_k} = \s{x_{k-1}x^*_{k-1}},\ k = 1,2,\ldots,$
with \s{x^*_{k-1}} identical to \s{x_{k-1}} except in the last position,
where the letter $c_k > c_{k-1}$ replaces $c_{k-1}$.
Then \s{x_k} has length $n = (h\+ 1)m$,
where $m = 2^{k+1}$ is the number of ranges in \s{x_k}.
In Appendix 4 it is shown in Lemma~\ref{lemm-rm}
that $\s{x_k}$
is a worst-case input for Algorithm $\NSV^*$,
which requires $\O(n\log n)$ range matches in such cases.
Since $\PNSV^*$ compares two ranges in $\O(\sigma)$ time,
it therefore requires $\O(\sigma n\log n)$ time
in the worst case, thus $\O(n\log n)$ for constant $\sigma$.
In Appendix 4 we argue that $\NPNSV^*$ is also $\O(n\log n)$
in the worst case.
%Experimentally, the two methods are about the same
%on the strings we have tested.

\section{Experimental Results}
\label{sect-exp}
We have done preliminary tests on the algorithms described above,
including the two variants of $\NSV^*$.
The equipment used was an Intel(R) Core i3 at 1.8GHz and 4GB main memory
under a 64-bit Windows 7 operating system.
Figure~\ref{fig-binary} shows the results of exhaustive tests of the algorithms on all
binary strings of lengths 11--22,
with all but RDuval displaying linear-time behaviour.
MaxLyn and NP$\NSV^*$ are roughly equivalent in time requirement,
with NSVISA several times slower, P$\NSV^*$ perhaps 10 times slower.

We have also tested the linear average-case algorithms on much longer binary strings,
several megabytes in length,
both random and highly periodic \cite{FSS03}.
On random strings, P$\NSV^*$ and NP$\NSV^*$ are comparable in speed
and fastest by a factor of 2 or 3,
while on the periodic strings, MaxLyn has an advantage by approximately
the same margin.
More testing needs to be done,
especially on strings defined on larger alphabets,
but of the current collection,
it appears that the two new $\O(n\log n)$-time algorithms
are the algorithms of choice.
\begin{figure}[htbp]
\centering
\includegraphics*[scale = 0.6]{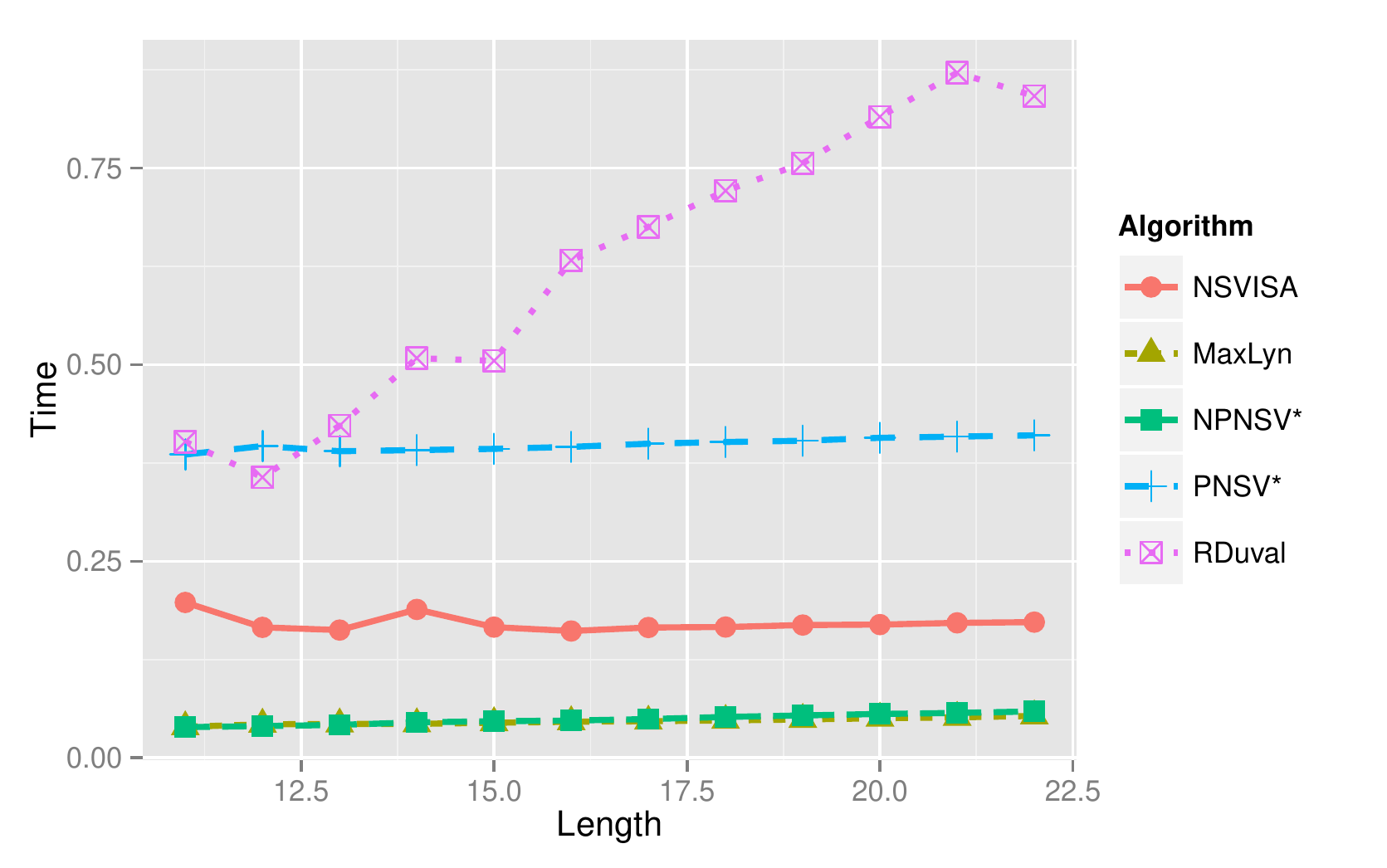}
\caption{Five algorithms compared on all binary strings of lengths $n \in 11..22$: the average processing time for each $n$ is given in $10^{-4}$ seconds.}
\label{fig-binary}
\end{figure}

\section{Future Work}
\label{sect-future}
There is reason to believe \cite{K14}
that the Lyndon array computation is less hard than
suffix array construction.
Thus the authors conjecture that there is a linear-time elementary
algorithm (no suffix arrays) to compute the Lyndon array.

\def\AAIM{International Conference on Algorithmic Aspects in Information \& Management}
\def\AJC{Australasian J.\ Combinatorics}
\def\AWOCA{Australasian Workshop on Combinatorial Algs.}
\def\CPM{Annual Symp.\ Combinatorial Pattern Matching}
\def\COCOON{Annual International Computing \& Combinatorics Conference}
\def\FOCS{IEEE Symp.\ Found.\ Computer Science}
\def\AESA{Annual European Symp.\ on Algs.}
\def\LATA{Internat.\ Conf.\ on Language \& Automata Theory \& Applications}
\def\IWOCA{Internat.\ Workshop on Combinatorial Algs.}
\def\AWOCA{Australasian Workshop on Combinatorial Algs.}
\def\STACS{Symp.\ Theoretical Aspects of Computer Science}
\def\ICALP{Internat.\ Colloq.\ Automata, Languages \& Programming}
\def\IJFCS{Internat.\ J.\ Foundations of Computer Science}
\def\ISAAC{Internat.\ Symp.\ Algs.\ \& Computation}
\def\SWAT{Scandinavian Workshop on Alg.\ Theory}
\def\PSC{Prague Stringology Conf.}
\def\ALG{Algorithmica }
\def\CSUR{ACM Computing Surveys}
\def\DAM{Discrete Applied Math.}
\def\FI{Fundamenta Informaticae}
\def\IPL{Inform.\ Process.\ Lett.\ }
\def\InfComp{Inform.\ \& Computation}
\def\IS{Inform.\ Sciences}
\def\JACM{J.\ Assoc.\ Comput.\ Mach.\ }
\def\JCTA{J.\ Combinatorial Theory, Series A}
\def\CACM{Commun.\ Assoc.\ Comput.\ Mach.\ }
\def\MCS{Math.\ in Computer Science}
\def\SICOMP{SIAM J.\ Computing}
\def\SIDMA{SIAM J.\ Discrete Math.\ }
\def\JCB{J.\ Computational Biology\ }
\def\JA{J.\ Algorithms }
\def\JDA{J.\ Discrete Algorithms }
\def\JALC{J.\ Automata, Languages \& Combinatorics\ }
\def\SODA{ACM-SIAM Symp.\ Discrete Algs.\ }
\def\SPE{Software, Practice \& Experience\ }
\def\TCJ{The Computer Journal}
\def\TCS{Theoret.\ Comput.\ Sci.\ }

\bibliographystyle{plain}
\bibliography{references}

\section*{Appendix 1}
The following result justifies the strategy employed in Algorithm MaxLyn (Figure~\ref{fig-maxlyn}):
\begin{lemm}
\label{lemm-maxlyn}
Suppose that for some position $i$ in a Lyndon word $\s{x}[1..n]$,
$k = \pi[i] \ge 2$.
Then for every $j \in i+1\dd i+k-1$, $\pi[j] \le i+k-j$.
\end{lemm}
\begin{proof}
The result certainly holds for $i+k = n+1$, so we consider $i+k \le n$.
Assume that for some $j \in i+1\dd i+k-1$,
$\pi[j] > i+k-j$.
It follows that
\begin{equation}
\label{eq1}
\s{x}[1\dd i+k-j+1] = \s{x}[j\dd i+k],
\end{equation}
while $\s{x}[j-i+1\dd k] = \s{x}[j\dd i+k-1]$.
Since $\s{x}$ is Lyndon, therefore $\s{x}[1+k] \prec \s{x}[i+k]$, and so we find that
\begin{equation}
\label{eq2}
\s{x}[j-i+1\dd 1+k] \prec \s{x}[j\dd i+k].
\end{equation}
From (\ref{eq1}) and (\ref{eq2}) we see that
$\s{x}[1..k+1]$ has suffix $\s{x}[j-i+1..k+1]$ satisfying
$\s{x}[j-i+1..k+1] \prec \s{x}[1..i+k-j+1]$,
%$\s{x}[1..k+1]$ has suffix $\s{x}[j-i+1..k+1] \prec \mbox{ prefix } \s{x}[1..i+k-j+1]$,
contradicting the assumption that \s{x} is Lyndon.  \qed
\end{proof}

\section*{Appendix 2}
Here we prove Theorem~\ref{thrm-main} that justifies Algorithm~\ref{fig-3line}:
\begin{thrm}
\label{thrm-main}
For a given string $\s{x} = \s{x}[1..n]$ on alphabet $\Sigma$,
totally order by $\prec$,
let $\ISA = \ISA_{\s{x}}^{\prec}$.
Then for every $i \in 1..n$, the substring $\s{x}[i..j]$
is a longest Lyndon factor with respect to $\prec$
if and only if
\begin{itemize}
\item[$(a)$]
for every $h \in i\+ 1..j$, $\ISA[j] < \ISA[h]$; and
\item[$(b)$]
either $j = n$ or $\ISA[j\+ 1] < \ISA[i]$.
\end{itemize}
\end{thrm}

The following well-known result is needed to prove Lemma~\ref{lemm-AB}:

\begin{lemm}[Duval, Lemma~1.6, \cite{D83}]\label{duval}
Suppose $\s{x} \in \Sigma^+$, where $\Sigma$
is an alphabet totally ordered by $\prec$.
Let $\s{x} = \s{u}^r\s{u_1}b$,
where \s{u} is nonempty, $r \ge 1$, \s{u_1} a possibly empty proper prefix of \s{u},
and the letter $b \ne \s{u}[|\s{u_1}|\+ 1]$.
\begin{my_itemize}
\item[$(a)$] If $b\prec \s{u}[|\s{u_1}|{+}1]$, then \s{u} is a longest Lyndon prefix
of \s{xy} for any \s{y};
\item[$(b)$] if $b\succ \s{u}[|\s{u_1}|{+}1]$, then \s{x} is Lyndon with respect to $\prec$.
\end{my_itemize}
\end{lemm}

For a given string $\s{x}[1..n]$, let $\s{s}_{\s{x}}(i) = \s{x}[i..n]$
denote the suffix of \s{x} beginning at position $i$.
When clear from context we write just $\s{s}(i)$.

\begin{lemm}\label{lemm-AB}
Consider a string $\s{x}=\s{x}[1\dd n]$ over alphabet $\Sigma$
totally ordered by $\prec$.
Let $\s{x}[i\dd j]$ be the longest
Lyndon factor of $\s{x}$ starting at $i$.
Then $\s{s}_{\s{x}}(i)\prec \s{s}_{\s{x}}(k)$ for every $k \in i\+ 1..j$
and either $j=n$  or $\s{s}_{\s{x}}(j{+}1) \prec \s{s}_{\s{x}}(i)$.
\end{lemm}

\begin{proof}
Because $\s{x}[i\dd j]$ is Lyndon, therefore for any $i<k\leq j$, $\s{x}[i\dd j]\prec \s{x}[k\dd j]$ and so
$\s{s}(i)\prec s(k)$. If $j=n$, we are done. So we may assume
$j < n$, and we want to show that $\s{s}(j{+}1)\prec \s{s}(i)$.
Suppose then that $\s{s}(j{+}1)\not\prec \s{s}(i)$.
Since $\s{s}(i)$ and $\s{s}(j{+}1)$ are distinct,
it follows that
$\s{s}(i)\prec \s{s}(j{+}1)$.
If we let $d=lcp(\s{s}(i),\s{s}(j{+}1))+1$, two cases arise:
\begin{my_itemize}
\item[(a)] $0\leq d\leq j-i$.\\
Here $i\leq i+d\leq j$. Thus $\s{x}[i\dd i{+}d{-}1]=\s{x}[j{+}1\dd j{+}d]$
and $\s{x}[i{+}d]\prec \s{x}[j{+}1{+}d]$, and so for $j<k\leq j{+}1{+}d$,
$\s{x}[i\dd j{+}1{+}d]\prec \s{x}[k\dd j{+}1{+}d]$. Since
$\s{x}[i\dd j]$ is Lyndon, $\s{x}[i\dd j]\prec \s{x}[k\dd j]$ and
so $\s{x}[i\dd j{+}1{+}d]\prec \s{x}[k\dd j{+}1{+}d]$ for any $i<k\leq j$.
Thus $\s{x}[i\dd j{+}1{+}d]$ is Lyndon,
contradicting the assumption that
$\s{x}[i\dd j]$ is the longest Lyndon factor starting at $i$.
\item[(b)] $0 < j-i\leq d$.\\
Let $d = r(j-i)+d_1$, where $0\leq d_1<j-i$. Then $r\geq 1$ and
$\s{x}[i\dd j{+}1{+}d]=\s{u}^r\s{u_1}b$ where
$\s{u}=\s{x}[i\dd j]$, $$\s{u_1} = \s{x}[j{+}r(j{-}i){+}1\dd j{+}r(j{-}i){+}d_1{-}1] = \s{x}[j{+}r(j{-}i){+}1\dd j{+}d{-}1]$$ is a prefix of
$\s{x}[i\dd j]$, and $\s{x}[i{+}d]\prec \s{x}[j{+}1{+}d]$, so that
by Lemma~\ref{duval}~(b),
$\s{x}[i\dd j{+}1{+}d]$ is Lyndon, contradicting the assumption that $\s{x}[i\dd j]$
is the longest Lyndon factor starting at $i$.
\end{my_itemize}
Thus $\s{s}(j{+}1)\prec \s{s}(i)$, as required. \qed
\end{proof}

Lemma~\ref{maxlyn-suf} describes the property of being a longest
Lyndon factor of a string \s{x} in terms of relationships between corresponding suffixes.

\begin{lemm}\label{maxlyn-suf}
Consider a string $\s{x} = \s{x}[1\dd n]$ over an alphabet $\Sigma$
with an ordering $\prec$.
A substring $\s{x}[i\dd j]$ is a longest
Lyndon factor of \s{x} with respect to $\prec$ if and only if
$\s{s}_{\s{x}}(i)\prec \s{s}_{\s{x}}(k)$ for every $k \in i+1..j$ and
either $j=n$  or $\s{s}_{\s{x}}(j{+}1) \prec \s{s}_{\s{x}}(i)$.
\end{lemm}

\begin{proof}
Let (A) denote \emph{$\{\s{x}[i\dd j]$ is a longest Lyndon factor of $\s{x}\}$} and
let (B) denote
\emph{$\{\s{s}(i)\prec \s{s}(k)$ for any $1\leq k\leq j$  and
$\s{s}(j{+}1)\prec \s{s}(i)\}$}.
Then (A) $\Rightarrow$ (B) follows from Lemma~\ref{lemm-AB}, so we need to
prove that (B) $\Rightarrow$ (A).

\medskip
\noindent
Suppose then that (B) holds,
and let $\s{x}[i\dd k]$ be a longest Lyndon factor of \s{x} starting at position $i$.
If $k<j$, then by Lemma~\ref{lemm-AB}, $\s{s}(k{+}1)\prec \s{s}(i)$, a contradiction
since $k{+}1\leq j$. If $k>j$, then by Lemma~\ref{lemm-AB},
$\s{s}(i)\prec \s{s}(j{+}1)$ because $j{+}1\leq k$, which again gives us a contradiction.
Thus $k=j$ and $\s{x}[i\dd j]$ is a longest Lyndon factor of \s{x}.  \qed
\end{proof}
Now we reformulate Lemma~\ref{maxlyn-suf} in terms
of the inverse suffix array $\ISA$ of \s{x} using the relationship that $\s{s}(i)\prec s(j) \Longleftrightarrow \ISA[i] < \ISA[j]$,
thus yielding Theorem~\ref{thrm-main}, as required.
Hence the Lyndon array can be computed in a simple three-step
algorithm, as shown in Figure~\ref{fig-3line},
that executes in $\theta(n)$ time and uses only one additional
array of integers.

\section*{Appendix 3}
Here we describe a simple data structure
that yields an alternative approach to Algorithm $\NSV^*$,
based on the comparison of longest Lyndon factors as
described in Lemma~\ref{maxlyn-suf}.
The \itbf{dictionary of basic factors} \cite{Crochemore2007,CR02}
of string $\s{x}[1..n]$ consists of a
sequence of arrays $\mathcal D_t, 0 \leq t \leq \log n$.
The array $\mathcal D_t$ records information about factors of $\s{x}$ of length $2^t$ ---
that is, the basic factors.
In particular, $\mathcal D_t[i]$
stores the rank of $\s{x}[i..i + 2^t - 1]$, so that
$$\s{x}[i..i + 2^t - 1] \preceq \s{x}[i..i + 2^t - 1] \Leftrightarrow \mathcal D_t[i] \leq  \mathcal D_t[i].$$
This dictionary requires $O(n \log n)$ space and
can be constructed in $O(n \log n)$ time as follows.
$\mathcal D_0$ contains information about consecutive symbols of $\s{x}$
and hence can be computed in $O(n \log n)$ time
by sorting all the symbols appearing in $\s{x}$ and mapping them to numbers from 1 and onward.
Once $\mathcal D_t$ is computed, we can easily compute $\mathcal D_{t+1}$ by
spending $O(n)$ time on a radix sort,
because $u[i..i+2^{t+1} -1]$ is in fact
a concatenation of the factors $u[i..i+2^t -1]$ and $u[i+2^t..i+2^{t+1} -1]$.

Once this dictionary is computed, we can compare any two factors by comparing two appropriate overlapping basic factors (i.e., factors having length power of two), which is done by checking the corresponding $\mathcal D$ array from the dictionary. This will require constant time and hence each suffix-suffix comparison can be done in constant time.

\section*{Appendix 4}
\begin{lemm}
\label{lemm-rm}
Let $RM^{(h)}(k)$ denote the number of range matches
needed by Algorithm $\NSV^*$ to compute $\s{\lambda}_{\s{x_k}}$ of length $n = (h\+ 1)m$,
where $m = 2^{k+1}$ is the number of ranges in \s{x_k}.
Then $RM^{(h)}(k) = m(\log_2 m\- 1) \+ 1 \in \Theta(n\log n)$.
\end{lemm}
\begin{proof}
Consider the rightmost two ranges $\s{s_0} = a^hba^hc_k$ of \s{x_k}.
$\NSV^*$ requires one range match to discover that $a^hb < a^hc_k$,
which we may denote by the vector $(1,0)$ that associates the one match
with the leftmost of the two ranges being compared.
Similarly, with the rightmost four ranges $\s{s_1} = a^hba^hc_0a^hba^hc_k$ of \s{x_k}
we may associate the vector $(2,2,1,0)$, counting a maximum two more range matches
performed by $\NSV^*$ on each of $a^hb$ and $a^hc_0$ with ranges to their right.
Observe that as the vector is extended to the left, the existing elements are unchanged.
Now consider the four ranges $a^hba^hc_0a^hba^hc_{k-1}$
that form the prefix of \s{s_2} on the left of \s{s_1}.
It is easy to see that the maximum number of range matches associated with the start positions
of these four ranges
can be counted $(3,3,3,3)$,
thus extending the vector to $(3,3,3,3,2,2,1,0)$.
The next eight positions on the left will yield a maximum $(4,4,4,4,4,4,4,4)$
range matches,
and so on, until the beginning of \s{x_k} is reached.
Thus
\begin{eqnarray*}
RM^{(h)}(k) & = & \sum_{j=0}^k (j\+ 1)2^j \\
& = & \sum_{j=1}^k j2^j + \sum_{j=0}^k 2^j \\
& = & (k(2^{k+2}\- 2^{k+1})\- 2^{k+1}\+ 2) + (2^{k+1}\- 1) \mbox{\ \ \cite[p.\ 33]{KNUTH73I}} \\
& = & k2^{k+1}\+ 1,
\end{eqnarray*}
and so $RM^{(h)}(k) = m(\log_2 m \- 1) \+ 1$, as required.  \qed
\end{proof}
Consider the vectors formed in the proof of Lemma~\ref{lemm-rm}
that count range matches.
Each position in the righthand vector $(1,0)$
is clearly largest possible
over all selections of ranges,
as are the preceding positions $(2,2)$.
Similarly, none of the values in $(3,3,3,3)$ can possibly be greater than 3:
in each case the three matches result from inequalities in the last positions
of the ranges being matched.
We see that in fact the vector corresponding to \s{x_k} must be maximal, and so,
when each range match requires constant time
(proportional to $\sigma$):
\begin{lemm}
\label{lemm-pnsv}
Algorithm P$\NSV^*$ computes $\s{\lambda}_{\s{x}}$ in $\O(n\log n)$ time for all \s{x}.
\end{lemm}
Consider now the execution of $\NPNSV^*$ on the strings $\s{x_k}$.
Instead of one comparison per range match by $\PNSV^*$,
now $h\+ 1$ letter comparisons are required.
For $h = 1$, the number of comparisons per range match
is therefore $2$, a multiple by a constant factor,
thus still linear time per match.
For arbitrary $h > 2$, the number of comparisons increases
by a factor of $h$, but at the same time range length
(and therefore string length) increases by a factor of $(h\+ 1)/2$,
so that still $O(n\log n)$ ranges are processed
in $O(n\log n)$ time.
Thus
\begin{lemm}
\label{lemm-npnsv}
Algorithm $\NPNSV^*$ computes $\s{\lambda}_{\s{x}}$
in $\O(n\log n)$ time for all \s{x}.
\end{lemm}

\end{document}